\documentclass{article}

\usepackage{arxiv}

\usepackage[utf8]{inputenc} 
\usepackage[T1]{fontenc}    
\usepackage{hyperref}       
\usepackage{url}            
\usepackage{booktabs}       
\usepackage{amsfonts}       
\usepackage{nicefrac}       
\usepackage{microtype}      
\usepackage{lipsum}

\usepackage{amsfonts}
\usepackage{amsmath}
\usepackage{subcaption}
\usepackage{graphicx}
\usepackage{amssymb}
\usepackage{pifont}
\usepackage[inkscapelatex=false]{svg}
\usepackage{mathtools}
\usepackage{xcolor}
\usepackage{newtxtext,newtxmath}
\usepackage{lipsum}
\usepackage{multirow}
\usepackage{cite}
\usepackage{float}
\usepackage{algorithm}
\usepackage{algpseudocode}
\usepackage{stmaryrd}
\usepackage{algorithm}
\usepackage{algpseudocode}

\newcommand{\R}{{\mathbb{R}}}

\newcommand{\N}{{\mathbb{N}}}

\newcommand{\X}{{\mathbf{X}}}

\newcommand{\cmark}{\ding{51}}
\newcommand{\xmark}{\ding{55}}

\newtheorem{theorem}{Theorem}[section]
\newtheorem{assumption}{Assumption}

\newtheorem{definition}[theorem]{Definition}
\newtheorem{lemma}[theorem]{Lemma}
\newtheorem{remark}[theorem]{Remark}
\newenvironment{proof}{\paragraph{Proof:}}{\hfill$\square$}
\newtheorem{problem}[theorem]{Problem}
\usepackage{graphicx} 

\title{Fast and Sample Efficient Safety Verification via Extreme Learning Machine 
}

\author{
 Ahan Basu \\
  Centre for Cyber-Physical Systems\\
  Indian Institute of Science, Bengaluru, India\\
  \texttt{ahanbasu@iisc.ac.in} \\
   \And
 Mahathi Anand \\
  Chair of Robotics and System Intelligence \\ Munich Institute of Robotics and Machine Intelligence \\ Technical University of Munich, Germany\\
  \texttt{mahathi.anand@tum.de} \\
  \And
 Pushpak Jagtap \\
  Centre for Cyber-Physical Systems\\
  Indian Institute of Science, Bengaluru, India\\
  \texttt{pushpak@iisc.ac.in} \\
}

\begin{document}

\maketitle

\begin{abstract}
Deep learning methods like neural networks have greatly simplified the computation of safety certificates for complex nonlinear systems with unknown dynamics. However, due to the data-driven nature of these certificates and the complex architecture of neural networks, computation time as well as robustness guarantees across unseen data remain a challenge. This work aims to formally verify safety properties of discrete-time unknown systems by synthesizing extreme learning machine (ELM)-based barrier certificates. Compared to neural network counterparts, this approach greatly improves convergence guarantees and computational time due to its architectural simplicity and the convex nature of the underlying optimization problem. By minimizing the Lipschitz constant of the candidate barrier, we present a grid-based sampling technique to formally verify its validity using the minimum number of samples required. We demonstrate through numerical examples the effectiveness of our approach, and compare with traditional deep-learning based certificate synthesis to highlight its benefits.    
\end{abstract}

\section{Introduction}
Safety has become an utmost priority for complex safety-critical systems such as autonomous vehicles, robots, drones, etc. Barrier certificates provide an effective mechanism to formally guarantee the safety of the system. A barrier certificate is a real-valued function defined over the state space of the system such that it acts as a \textit{barrier} between the safe and unsafe regions. Correspondingly, the existence of such a function provides sufficient conditions for the system trajectories to avoid unsafe configurations. Barrier certificates were first proposed in the context of nonlinear systems~\cite{prajna2004safety} have since been used in the context of hybrid systems \cite{prajna2007framework} and stochastic systems \cite{jagtap2020formal, clark2021control}, among many others.

While it is straightforward to compute safe controllers given a barrier certificate~\cite{ames2019control}, synthesizing a barrier certificate that yields non-conservative results is a challenging problem. Traditionally, one requires to parameterize the barrier certificate with a specified template (e.g. polynomial functions), and then compute the corresponding coefficients via numerical optimization techniques such as sum-of-squares optimization \cite{parrilo2003semidefinite} or SMT solvers \cite{de2011satisfiability}. However, these methods have a high computational complexity, especially when dealing with complex, nonlinear system dynamics.  

Another fundamental challenge in synthesizing barrier certificates is the reliance on an accurate mathematical model of the system dynamics, which is often unavailable in practical applications. To overcome this limitation, several recent works have proposed learning barrier certificates directly from data using neural networks \cite{zhao2020synthesizing, mathiesen2022safety, dawson2023safe, anand2023formally} by leveraging the universal approximation property of neural networks \cite{barron1994approximation}. Nevertheless, existing DNN-based approaches remain computationally expensive because training requires repeated backpropagation while enforcing safety constraints over large regions of the state space, resulting in poor scalability. Moreover, their performance is sensitive to hyperparameter initialization and optimization settings, which can lead to slow or suboptimal convergence. Formal verification further compounds the computational burden, as it is commonly formulated as mixed-integer linear programming (MILP), satisfiability modulo theories (SMT), or semidefinite programming (SDP) problems whose computational complexity scales poorly with network size. 

To alleviate these issues, we propose to use extreme learning machines (ELM)~\cite{huang2006extreme}, a special case of single-layer neural networks whose hidden node weights are randomly assigned, while output layer weights are learned. The resulting convex programming admits a closed-form solution, eliminating iterative backpropagation and significantly reducing the training complexity \cite{zhou2024physics, markowska2021extreme}. In this work, we solve the safety verification problem for unknown discrete-time dynamical systems using ELM-based barrier certificates.To the best of our knowledge, no prior work exists that directly synthesizes barrier certificates for unknown dynamical systems using ELM without learning the system dynamics. In particular, we characterize barrier certificates as ELMs, and learn them using suitable data collected from the state sets. The sufficient conditions required to be satisfied by the barrier certificate are reformulated into appropriate loss functions. Since the hidden layers of the ELM are fixed, the barrier certificate conditions are affine in the weights of the output layer, and as a result, the loss function is also rendered affine and hence convex. However, the loss function is only evaluated over the training data, and as a result, the learned certificate may not be valid over unseen data. Therefore, to improve robustness, we bound the Lipschitz constant of the ELM through weight normalization and regularization during training. Moreover, to formally validate the certificates, a post-hoc verification is conducted using the Lipschitz bounds on the network. Specifically, we derive a Lipschitz constant-based validity condition whose satisfaction over finitely many grid-based data samples leads to the satisfaction of barrier certificate conditions over the entire state set. Finally we validate the efficacy of our approach with two case studies.
We also compare our method with existing neural methods, showcasing a significant reduction in computation time and samples required to provide the formal guarantee.

\section{Preliminaries and Problem Formulation}

\textit{Notations}: The symbols $\N$, $ \N_0$, $\R$, $\R^+$, and $\R_0^+$ denote the set of natural, non-negative integers, real, positive real, and nonnegative real numbers, respectively. 
The space of real matrices with $n$ rows and $m$ columns is denoted by $\R^{n\times m}$. The set of column vectors with $n$ rows is denoted by $ \R^n$.
{The Euclidean norm of a vector is represented by $|\cdot |$ while the spectral norm of a matrix is denoted by $\lVert \cdot \rVert$.}
For $a, b \in \N_0$ with $a \leq b$, the closed interval in $\N_0$ is denoted by $[a; b]$. 
Given a matrix $M\in\R^{n\times m}$, $M^\top$ represents the transpose of $M$. 
{For all $x, y \in \R^n$, the vector inequality $x \preceq y$ represents $x_i \leq  y_i$ for all $i \in [1;n]$.}

\subsection{System Definition}

\begin{definition}\label{def:system}
Consider a discrete-time dynamical system represented by the tuple $\Xi = (\X, \X_0, f)$, where $\X \subseteq \R^n$ is the state space of the system, $\X_0 \subset \X$ denotes set of initial states of the system, and $f: \X \rightarrow \X$ describes the state evolution via the following difference equation:
\begin{equation}\label{eq:act_system}
        \mathsf{x}(k+1) = f(\mathsf{x}(k)), \quad \forall k \in \N_0,
\end{equation}
with $\mathsf{x}(k) \in \X$ is the state of the system at $k$-th time instant.
\end{definition}

Given an initial state $x \in \X_0$, we denote the state sequence or the discrete-time trajectory of the system by $\mathbf{x}_x$, where $\mathbf{x}_x(k)= \mathsf{x}(k)$ with $\mathbf{x}_x(0)= x$. In this paper, the map $f$ is considered to be unknown, but we assume the availability of a black-box model such that, given the current state to the system, the next state of the system is accessible. As a result, we can collect finitely many samples of the system as $\{(x^{(0)}, f(x^{(0)})), \ldots, (x^{(N)}, f(x^{(N)}))\}$. The goal of the paper is to verify safety properties of the system, i.e., given an initial set $\X_0 \subset \X$ and unsafe set $\X_U \subset \X$, we want to determine if $\mathbf{x}_x(k) \in \X \setminus \X_U$, for all $x \in \X_0$ and $k \in \N$. To do so, we first recall the notion of barrier certificates as follows.

\subsection{Barrier Certificates}
{We propose the following lemma to correlate the safety of dynamical systems and the existence of barrier certificates.}
\begin{lemma} \cite{jagtap2020compositional}
    Given a discrete-time dynamical system $\Xi = (\X, \X_0, f)$ with initial set of states $\X_0 \subset \X$ and a set of unsafe states $\X_U \subset \X$ such that $\X_0 \cap \X_U = \emptyset$. If there exists a barrier certificate $B:\X \rightarrow \R$ satisfying the following conditions:
    \begin{subequations} \label{eq:CBC}
    \begin{align}
        B(x) \!&\leq\! 0, \quad \forall x \in \X_0, \label{eq:CBC_leq} \\
        B(x) \!&>\! 0, \quad \forall x \in \X_U, \label{eq:CBC_geq} \\
        B(f(x)) \!-\! B(x) \!&\leq\! 0, \quad \forall x \in \X, \label{eq:CBC_diff}
    \end{align}
    \end{subequations}
   then for every initial condition $x \in \X_0$ and for all $k \in \N$, the corresponding state trajectory $\mathbf{x}_x(k) \in \X \setminus \X_U$.
\end{lemma}

The connection between system safety and the existence of a BC is straightforward. The zero-level set, $B(x)=0$, defines a boundary separating the safe and unsafe regions. According to condition \eqref{eq:CBC_leq}, any initial state $x_0 \in \X_0$ inherently satisfies $B(x_0) \le 0$, while condition \eqref{eq:CBC_geq} ensures that $B(x)$ is positive inside the unsafe set. Condition \eqref{eq:CBC_diff} ensures that $B(x)$ is non-increasing along system trajectories and consequently, trajectories cannot cross the boundary and move into the unsafe region, thereby guaranteeing safety. For systems with known dynamics, BCs are typically synthesized by fixing a functional template (e.g., a polynomial) and searching for its coefficients using techniques such as sum-of-squares (SOS) optimization or satisfiability modulo theory (SMT) solvers \cite{jagtap2020formal}. However, for an unknown system, it is intractable to compute the BC due to the presence of the map $f$ in condition \eqref{eq:CBC_diff}. Some recent works parameterize the BC as a neural network \cite{anand2023formally, zhao2020synthesizing} and learn it in a data-driven fashion without any knowledge of the map $f$. Most of the neural-network based approaches rely on backpropagation to optimize a non-convex objective, making convergence sensitive to the weight of the neural network initialization while offering no guarantee of attaining a global optimum. Motivated by these limitations, we propose an ELM-based BC formulation in the next section that recasts the synthesis problem as a convex program, enabling feasible solutions with polynomial-time computational complexity \cite[Ch.~11]{boyd2004convex}.

\subsection{Extreme Learning Machine}
An extreme learning machine (ELM) is a special case of a single-hidden-layer feedforward neural network given by 
\begin{align}\label{eq:elm}
    B(x; \beta) = \beta^\top \phi(Wx + b),
\end{align}
where $W\! \in \! \R^{p \times n}, b \! \in \! \R^p, \beta \! \in \! \R^p, p \! \in \! \N$ and $\phi$ is an activation function applied elementwise. While any activation function\footnote{To obtain universal approximation guarantees, a non-polynomial activation function is required.} is theoretically justified, we exclusively use any Lipschitz-bounded activation function (for example, ReLU, Sigmoid, Tanh, etc.), a property that is useful in providing formal guarantees. Note that the ELM consists of $p$ hidden neurons, with hidden layer weights $W$ and biases $b$ pre-selected randomly (i.e., are not trainable), while the output weights $\beta$ are unknown and require to be learned. 

\subsection{Problem Formulation}
We are now ready to state the main problem of the paper.
\begin{problem}\label{prob}
    Given a discrete-time dynamical system $\Xi = (\X, \X_0, f)$ as in \eqref{eq:act_system} with the initial set $\X_0 \subset \X$, unsafe set $\X_U \subset \X$ and an unknown map $f$, synthesize a suitable barrier certificate parameterized by an ELM as in \eqref{eq:elm} that verifies the system starting within $\X_0$ never reaches $\X_U$.
\end{problem}

{Since the ELM architecture is such that hidden layer weights are pre-selected, the equation~\eqref{eq:elm} is linear in the parameters $\beta$, resulting in a candidate barrier certificate that is also linear in the unknowns. As a result, conditions~\eqref{eq:CBC} are also rendered affine in $\beta$, and can be quickly optimized by solving a convex program (CP),} rather than a generic non-convex training problem. As a result, convergence to the global optima is possible when the problem is feasible.  In the following sections, we provide a sound methodology to learn ELM-based barrier certificates using only finitely many data samples while providing formal guarantees over the continuous state space.

\section{ELM-based Barrier Certificate Synthesis} \label{sec:train}

In this section, we first formulate the dataset generation procedure and the associated loss functions for synthesizing ELM-based barrier certificates. We then derive a validity condition ensuring that the synthesized certificate, trained on a finite sample set, remains valid over the entire continuous state space. We conclude by presenting the full training algorithm for the proposed framework.  

\subsection{Formulation of Loss Function}

First, we present the dataset generation and formulation of loss functions for ELM-based barrier certificates. The goal is to train a Lipschitz-bounded ELM-based barrier certificate, i.e., a function $B(x;\beta)$ such that $|B(x_1;\beta) - B(x_2; \beta) | \leq \mathcal{L}_{barr} |x_1 - x_2|$, for any $x_1, x_2 \in \X$ for some Lipschitz constant $\mathcal{L}_{barr}$. The bound $\mathcal{L}_{barr}$ is used for achieving formal guarantees, as detailed in the next subsection. The Lipschitz constant of the ELM is obtained as $\mathcal{L}_{barr} = \| W \| |\beta | \mathcal{L}_\phi$, where $\mathcal{L}_\phi$ is the Lipschitz constant of the activation function $\phi(\cdot)$ \cite{szegedy2013intriguing}. Note that most common activations like $\mathrm{ReLU}, \mathrm{Tanh}$, etc have a Lipschitz constant of $1$. 
However, since the weight $W \in \R^{p \times n}$ and bias $b \in \R^p$ are randomly chosen, the $\|W\|$ may be a large value. To evade this, we normalize the weight matrix as $\overline{W}: = W / \sigma_{\max}(W)$ where $\sigma_{\max}(\cdot)$ denotes the largest singular value of $W$, such that spectral norm of $\overline{W}$ satisfies $\lVert \overline{W}\rVert = 1$. Therefore, the structure of ELM \eqref{eq:elm} is modified as $B(x; \beta) = \beta^\top \phi(\overline{W} x + b)$.

To train the ELM barrier candidate, we collect a finite number of samples from the state space $\X$. First, we construct cover sets for the sets $\X_0, \X_U$, and $\X$ respectively. A cover of a set $\X$ is a set consisting of subsets of $\X$ whose union equals $\X$. The cover sets are constructed such that they consist of hyper-rectangles
\begin{align*}
    H_i (x_i, \varepsilon_i):=\{x \in \X | -\varepsilon_i \preceq x - x_i \preceq \varepsilon_i\},
\end{align*}
centered at grid points $x_i \in \X, \quad i \in [1;N]$ with $\varepsilon_i \in \R^n$ where $N$ denotes the number of cover subsets. Now, if we consider $\varepsilon = \max_i |\varepsilon_i|$, then for any $x \in \X$, there exists $x_i, i \in [1;N]$ such that $|x - x_i| \leq \varepsilon$. Therefore, by collecting the representative points $x_i$ for all $i \in [1;N]$, we create the dataset $\mathcal{X}$. To generate the datasets $\mathcal{X}_0$ and $\mathcal{X}_U$ corresponding to the initial set $\X_0$ and the unsafe set $\X_U$, we assign the existing representative points according to membership as $\mathcal{X}_0\!:=\!\{x_s | x_s \!\in\! \mathcal{X} \text{ and } \exists~x \!\in\! \X_0 \text{ s.t. }|x \!-\! x_s| \!\leq\! \varepsilon\}$ and $\mathcal{X}_U:=\{x_s | x_s \!\in\! \mathcal{X} \text{ and } \exists~x \!\in\! \X_U \text{ s.t. }|x \!-\! x_s| \!\leq\! \varepsilon \}$. We consider $\mathcal{X}_0 \cap \mathcal{X}_U = \emptyset$.

Now, we reformulate conditions~\eqref{eq:CBC} as a suitable loss function for training the ELM as:
\begin{subequations}\label{eq:sub-loss}
\begin{align}
    \mathsf{L}_1(\beta) &= \sum_{x_s\in\mathcal{X}_0}\max(0,(B(x_s; \beta) - \eta)), \\
    \mathsf{L}_2(\beta) &= \sum_{x_s\in\mathcal{X}_U}\max(0, (-B(x_s; \beta) + \delta - \eta)), \\
    \mathsf{L}_3(\beta) \!&= \!\sum_{x_s\in\mathcal{X}} \max(0, \!(B(f(x_s);\! \beta) \! -\! B(x_s;\! \beta)\! -\! \eta)) \label{eq:loss_diff}, 
\end{align}
\end{subequations}
where $\eta \leq 0$ is a slack variable that ensures robustness and aids in formal verification described in next subsection, and $\delta \in \R^+$ is a positive quantity to ensure the strict inequality as in \eqref{eq:CBC_geq}. Note that one can directly compute loss \eqref{eq:loss_diff} for any data point in $\mathcal{X}$ by simulating the black-box dynamical system during training. 
Finally, the actual loss function is the weighted sum of the sub-loss functions:
\begin{align}\label{eq:loss_CBC}
    \mathsf{L}(\beta) = c_1 \mathsf{L}_1(\beta) + c_2\mathsf{L}_2(\beta) + c_3 \mathsf{L}_3(\beta),
\end{align}
where $c_1, c_2, c_3 \! \in \! \R^+$. Now, since sub-loss functions are convex in $\beta$ and the sum of convex functions is also a convex function, the loss \eqref{eq:loss_CBC} is also a convex function. 

\begin{remark}
Since the part $\phi(\overline{W}x + b)$ is fixed due to fixed $\overline{W}$ and $b$, every condition in \eqref{eq:loss_CBC} effectively becomes a convex constraint with respect to $\beta$; thereby, the program becomes a simple convex program (CP). 
\end{remark}

Now we provide the following lemma to show how the BC conditions have been satisfied for the sampled data points from the state space upon the convergence of the loss~\eqref{eq:loss_CBC}.

\begin{lemma}\label{lem:loss_zero}
    Given an unknown discrete-time dynamical system $\Xi = (\X, \X_0, f)$ with $\X_0$ and $\X_U$ being the initial and unsafe sets. If the loss function in~\eqref{eq:loss_CBC} $\mathsf{L}(\beta)$ converges to zero with some given $\eta < 0$ and $\delta \in \R^+$ over the datasets $\mathcal{X}_0, \mathcal{X}_U, \mathcal{X}$, then the trained ELM $B(x; \beta)$ satisfies the conditions in~\eqref{eq:CBC} for the sampled data points. 
\end{lemma}
\begin{proof}
    The loss function $\mathsf{L}(\beta) = 0$ essentially implies all of its individual sub-loss functions $\mathsf{L}_i(\beta)=0$ for all $i=1,2,3$. Now, $\mathsf{L}_1(\beta)=0$ implies $B(x_s; \beta) \leq \eta$ which implies $B(x_s; \beta) \leq 0$ for all $x_s \in \mathcal{X}_0$. Similarly, whenever the other sub-loss functions are equal to zero, it implies the satisfaction of other barrier conditions over the sampled points, which completes the proof. 
\end{proof}

\subsection{Correctness Guarantees for ELM-Barrier Certificates} 
While Lemma~\ref{lem:loss_zero} encourages the satisfaction of conditions~\eqref{eq:CBC} over the sampled points, it does not ensure the satisfaction of the conditions beyond the training data. Therefore, the synthesized function $B(x; \beta)$ is only a candidate, but may not be a true safety certificate, and one needs to verify its validity a posteriori. To do this, we leverage the Lipschitz regularity of the network, recalling that the Lipschitz constant of the network bounds the output perturbations with respect to its input perturbations, and minimizing it thus improves its robustness with respect to unseen data $x \in (\X \setminus \mathcal{X})$. To ensure the satisfaction of conditions~\eqref{eq:CBC} robustly, we first raise the following assumption.

\begin{assumption}\label{assum:Lipschitz}
    The map $f$ is Lipschitz continuous with respect to $x$ with Lipschitz constant $\mathcal{L}_x$ in Euclidean norm that is known a priori. In case the Lipschitz constant is unknown, one can estimate it using data collected from the system, as described in \cite[Algorithm 1]{nejati2023formal}.
\end{assumption}

We utilize the Lipschitz bounds of the synthesized candidate certificates along with those of the system to provide formal validity guarantees for the certificate. We recall that the Lipschitz constant of the ELM-based barrier certificate is given by $\mathcal{L}_{barr} = \|\bar W\| \times |\beta| \times \mathcal{L}_{\phi}$, where $\|\bar W\| = 1$ and $\mathcal{L}_{\phi}$ is known a priori. Therefore, by minimizing the $|\beta|$, one can minimize $\mathcal{L}_{barr}$. To do so, we propose the underlying problem as a quadratic program (QP) formulated as:
\begin{align}\label{eq:CP}
    \min_{\beta} \quad |\beta|^2 \quad 
    \text{s.t.} \quad \mathsf{L}(\beta) = 0. 
\end{align}

We now provide the main theorem of the section that formally verifies that the synthesized ELM is a valid barrier certificate for the unknown discrete-time system. 

\begin{theorem}\label{th:verification}
    Given an unknown discrete time dynamical system $\Xi = (\X, \X_0, f)$ as in \eqref{eq:act_system} with $\X_0, \X_U \subset \X$ being the corresponding initial and unsafe sets. Suppose $B(x; \beta^*)$ is the trained ELM-barrier with optimized weight $\beta^*$ such that the loss function $\mathsf{L}(\beta^*) = 0$ over the training datasets $\mathcal{X}_0, \mathcal{X}_U$ and $\mathcal{X}$ using Algorithm~\ref{algo:NN_training} with a given choice of $\eta$ and $\varepsilon$. Then the function $B(x; \beta^*)$ is guaranteed to be a valid barrier certificate and the system is guaranteed to be safe under Assumption \ref{assum:Lipschitz} if the following condition holds:
    \begin{align}\label{eq:cond_ver}
        \eta + \mathcal{L}\varepsilon \leq 0,
    \end{align}
    where $\mathcal{L} = \mathcal{L}_{barr}(\mathcal{L}_x + 1)$.
\end{theorem}
\begin{proof}
    We start the proof with the candidate barrier certificate $B(x; \beta^*)$, which is obtained upon solving the CP in Algorithm~\ref{algo:NN_training}. As the CP is feasible, it is evident that individual loss functions in \eqref{eq:loss_CBC} are zero with respect to the sampled points $x_s$ from the state space, which implies:
    \begin{align*}
        B(x_s; \beta^*) &\leq \eta, \quad \forall x_s \in \mathcal{X}_0, \\
        -B(x_s; \beta^*) +\delta &\leq \eta, \quad \forall x_s \in \mathcal{X}_U, \\
        B(f(x_s); \beta^*) - B(x_s; \beta^*) &\leq \eta, \quad \forall x_s \in \mathcal{X}.
    \end{align*}
    Now we prove that for all $x \in \X$ if the condition~\eqref{eq:cond_ver} holds, the 
    \begin{align*}
        &\text{(a) $ \forall x \in \X_0, \exists x_s \in \mathcal{X}_0$. Therefore:}\\
        & B(x; \beta^*) = B(x; \beta^*) - B(x_s; \beta^*) + B(x_s; \beta^*) \leq \mathcal{L}_{barr}|x - x_s| + \eta \leq \eta + \mathcal{L}\varepsilon \leq 0. \\
        &\text{(b) $\forall x \in \X_U, \exists x_s \in \mathcal{X}_U$. Therefore:}\\
        &-B(x; \beta^*) + \delta = -B(x; \beta^*) + B(x_s; \beta^*)  - B(x_s; \beta^*) + \delta \leq \mathcal{L}_{barr}|x - x_s| + \eta \leq \eta + \mathcal{L}\varepsilon \leq 0. \\
        &\text{(c) $\forall x \in \X, \exists x_s \in \mathcal{X}$. Therefore:}\\
        & B(f(x); \beta^*) - B(x; \beta^*) = B(f(x); \beta^*) \!-\! B(f(x_s); \beta^*) + B(f(x_s); \beta^*) - B(x; \beta^*) + B(x_s; \beta^*) - B(x_s; \beta^*) \\
        & \leq \mathcal{L}_{barr}| f(x) - f(x_s)| + \mathcal{L}_{barr} |x - x_s| + \eta \leq \mathcal{L}_{barr} (\mathcal{L}_x + 1) \varepsilon + \eta \leq 0.  
    \end{align*}
    Therefore, if condition \eqref{eq:cond_ver} is satisfied by the ELM, which has been trained using the sampled data points, the ELM is a valid barrier certificate for the discrete-time system, thereby completing the proof.
\end{proof}

\begin{remark}
    The choice of $\eta$ can be close to zero if a small $\varepsilon$ is considered; however, that means the number of samples is higher and computation time increases. Refer to Table~\ref{tab:computation_elm} for details of computation time with respect to sample complexity.    
\end{remark}

We present the training procedure of the ELM-barrier as follows in Algorithm~\ref{algo:NN_training}. 

\begin{algorithm}
\caption{Training of the ELM-barrier}
\label{algo:NN_training}
\begin{algorithmic}[1]
    \Require Data set: $\mathcal{X}_0, \mathcal{X}_U, \mathcal{X}$ 
    \Ensure $B(x; \beta)$
    \State Select the hyperparameters $ \eta, c = [c_1, c_2, c_3]$.
    \State Initialize ELM by sampling random weights, bias, and trainable parameters $W, b$ and $\beta$, respectively. Normalize the weights to $\overline{W}$ such that $\lVert \overline{W} \rVert = 1$.
    \State Solve the CP~\eqref{eq:CP} to find the optimal $\beta $ denoted by $\beta^*$.
    \If{CP infeasible} \Return \textsc{Infeasible, Change Hyperparameters} \EndIf
    \State Compute $\mathcal{L}_{barr} = |\beta^*| \times \mathcal{L}_{\phi}$ and get the constant $\mathcal{L}$ according to Theorem~\ref{th:verification}.
    \If{Condition~\eqref{eq:cond_ver} is unsatisfied} \Return \textsc{Not Verified, Change $\eta$ or $\varepsilon$.} \EndIf
    \State \textbf{return} $B(x; \beta^*)$.
\end{algorithmic}
\end{algorithm}

\section{Case Studies}
We perform our proposed methodology on two different examples, the first one being a simple pendulum and the second one being a magnetic levitator system. All the computations and comparisons were performed in Python 3.10 on a Windows machine with an Intel Core i7-14700 CPU and 32 GB RAM. {We have used the HiGHS interior-point QP solver \cite{huangfu2018parallelizing} to solve the QP~\eqref{eq:CP}.}

\subsection{Simple Pendulum System}
For the first case study, we consider a simple undamped pendulum with the following discrete-time dynamics:
\begin{subequations}
\begin{align}
    \mathsf{x}_1(k+1) &= \mathsf{x}_1(k) + \tau \mathsf{x}_2(k), \\
    \mathsf{x}_2(k+1) &= \mathsf{x}_2(k) - \tau \frac{g}{l}\sin(\mathsf{x}_1(k)),
\end{align}
\end{subequations}
where $\mathsf{x}_1, \mathsf{x}_2$ denote the angle and the angular velocity of the pendulum while $g$ and $l$ denote the acceleration due to gravity and the length of the pendulum, respectively. The domain of the system is given as $\X = [-\pi/6, \pi/6]^2$ while the initial and safe regions are considered as $\X_0 = [-\pi/15, \pi/15]^2, \X_s = [-\pi/8, \pi/8]^2$ while the unsafe region is chosen to be $\X_u = \X \setminus \X_s$. Though we consider the model to be unknown, the Lipschitz constant is considered to be known as $\mathcal{L}_x = 1.1$. 

The goal is to verify using some barrier certificate that if the system starts within $\X_0$, it will be safe. To do so, we have considered the barrier certificate as ELM $B(x; \beta)$. For synthesizing the certificate, we have considered $\varepsilon = 0.00123$ and set the hidden neuron number as $p=50$. We deploy Algorithm~\ref{algo:NN_training} to train the ELM and using $\eta = -0.0001$, the network is assured to be formally verified with Lipschitz constant $\mathcal{L}_{barr} = 0.033$ as $\eta + \mathcal{L}\varepsilon = -0.0001 + 0.033\times2.1\times0.00123 = -0.000016 < 0$. The obtained barrier certificate and its level set with a sample trajectory of the system are shown in Figure \ref{fig:simp_pend}. In addition, we performed experimental comparisons to analyze the certificate synthesis and its training time with respect to the number of hidden neurons, which measures the representation flexibility of the ELM, and the number of samples which dictates the validity of the obtained certificates via~\eqref{eq:cond_ver}. As seen in Table~\ref{tab:computation_elm}, a large number of samples is required to ensure verification success, whereas over $50$ neurons is sufficient to synthesize a barrier certificate for the example.   

\begin{figure}[h]
\centering
\begin{minipage}{0.48\textwidth}
\centering
\captionof{table}{Computation time comparison for different hidden units and sample size}
\label{tab:computation_elm}
\resizebox{\linewidth}{!}{%
\begin{tabular}{|c|c|c|c|c|}
\hline
\begin{tabular}[c]{@{}c@{}}Hidden \\ Units\end{tabular} &
  \begin{tabular}[c]{@{}c@{}}Epsilon\\ ($\varepsilon$)\end{tabular} &
  \begin{tabular}[c]{@{}c@{}}Number of\\ Samples (N)\end{tabular} &
  \begin{tabular}[c]{@{}c@{}}Computation \\ Time (in sec.)\end{tabular} &
  Verified \\ \hline
\multirow{3}{*}{50}  & 0.00740 & $1\times10^4$  & 0.4   & \xmark \\ \cline{2-5}
  & 0.00370 & $4\times10^4$  & 2.1   & \xmark \\ \cline{2-5}
  & 0.00123 & $36\times10^4$ & 21.2  & \cmark \\ \hline
\multirow{3}{*}{100} & 0.00740 & $1\times10^4$  & 0.9   & \xmark \\ \cline{2-5}
  & 0.00370 & $4\times10^4$  & 4.3   & \xmark \\ \cline{2-5}
  & 0.00123 & $36\times10^4$ & 48.5  & \cmark \\ \hline
\multirow{3}{*}{200} & 0.00740 & $1\times10^4$  & 1.8   & \xmark \\ \cline{2-5}
  & 0.00370 & $4\times10^4$  & 7.3   & \xmark \\ \cline{2-5}
  & 0.00123 & $36\times10^4$ & 110.1 & \cmark \\ \hline
\end{tabular}}
\end{minipage}%
\hfill
\begin{minipage}{0.48\textwidth}
\centering
\includegraphics[width=0.95\linewidth]{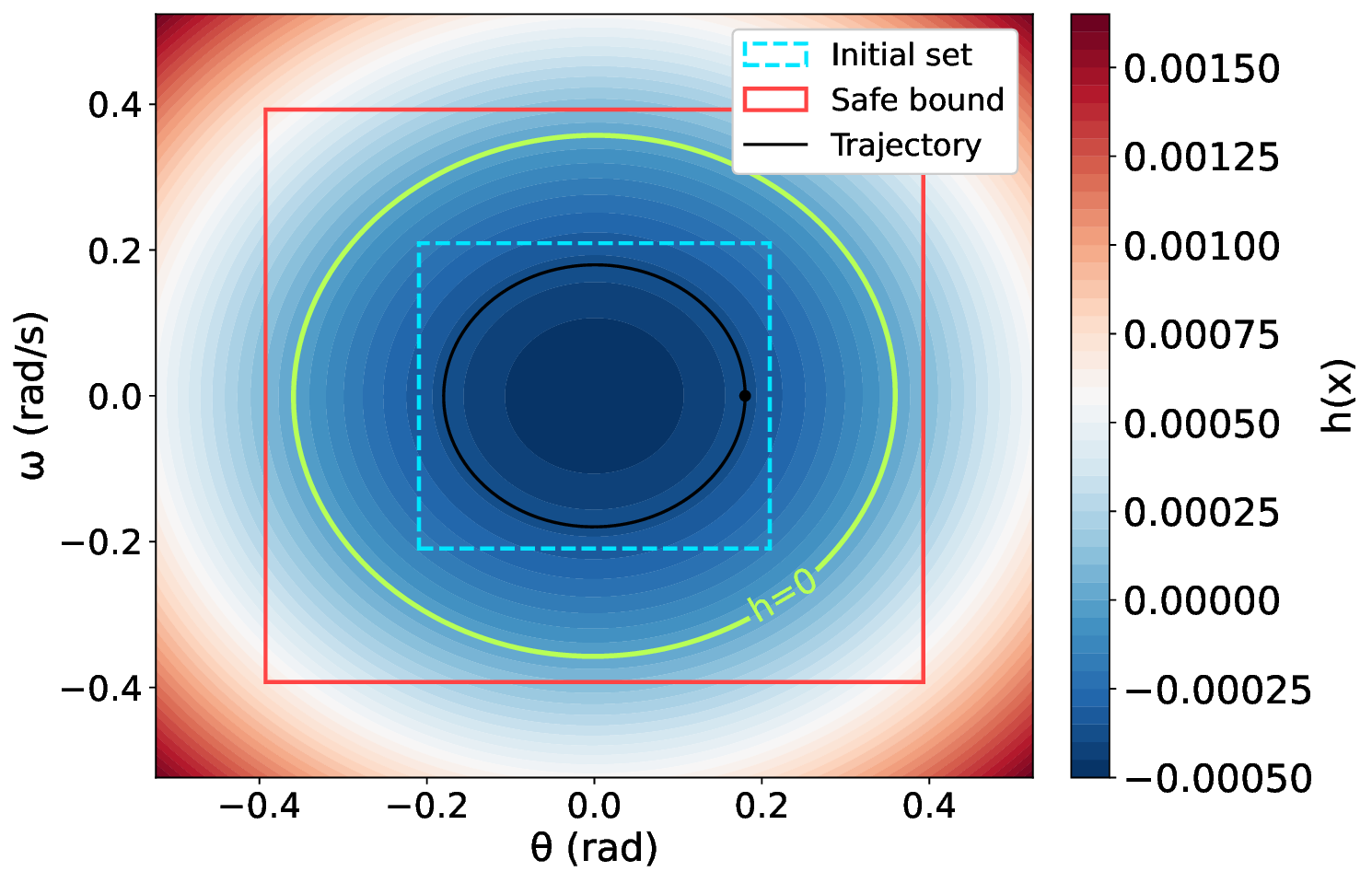}
\captionof{figure}{Simple Pendulum: Barrier Level Sets.}
\label{fig:simp_pend}
\end{minipage}
\end{figure}

\subsection{Magnetic Levitator System}
For the second case study, we consider a magnetic levitator system adapted from \cite{barie1996linear}, whose discrete-time dynamics is given by the following set of equations:
\begin{subequations}
\begin{align}
    \mathsf{x}_1(k\!+\!1) &=\! \mathsf{x}_1(k) + \tau \mathsf{x}_2(k), \\
    \mathsf{x}_2(k\!+\!1) &=\! \mathsf{x}_2(k) + \tau(g - \frac{b}{m} \frac{\mathsf{x}_3^2(k)}{\mathsf{x}_1^2(k)}), \\
    \mathsf{x}_3(k\!+\!1) &=\! \mathsf{x}_3(k) + \frac{\tau}{L}(u - R\mathsf{x}_3(k) + \frac{2b}{m}\frac{\mathsf{x}_2(k) \mathsf{x}_3(k)}{\mathsf{x}_1^2(k)}),
\end{align}
\end{subequations}
where $\mathsf{x}_1, \mathsf{x}_2, \mathsf{x}_3$ denote the position of the ball, velocity of the ball and induced current in the electromagnetic coil respectively while $m, g, b, R$ and $L$ denote the mass of the ball, acceleration due to gravity, magnetic force constant of the electromagnet, resistance and inductance of the coil, respectively. We consider an LQR tracking controller given by $u = u^* - K(x - x^*)$ where $x^*$ is the desired equilibrium position with $u^*$ being a nominal controller and $K$ is the gain. The domains of the system are considered as $\X = [85,115] \times [40, 48] \times [0.25, 0.35]$, while the initial and safe regions are considered as $\X_0 = [97.5, 102.5] \times [44,46] \times [0.28, 0.3], \X_s = [90, 110] \times [42,47] \times [0.26, 0.33]$ while the unsafe region is chosen to be $\X_u = \X \setminus \X_s$. Though we consider the model to be unknown, the Lipschitz constant is considered to be known as $\mathcal{L}_x = 5.964$. 

\begin{table*}[t]
\centering
\caption{Qualitative comparison of proposed ELM-based technique with other methods}
\label{tab:quality-elm}
\resizebox{0.95\textwidth}{!}{
\begin{tabular}{|c|c|c|c|c|c|c|}
\hline
Methods             & \begin{tabular}[c]{@{}c@{}}Unknown \\ Dynamics\end{tabular} & \begin{tabular}[c]{@{}c@{}}Convex \\ Program\end{tabular} & \begin{tabular}[c]{@{}c@{}}Backpropagation \\ Independency \end{tabular} & Scalability & \begin{tabular}[c]{@{}c@{}}General \\ Nonlinear \\ Dynamics \end{tabular} & \begin{tabular}[c]{@{}c@{}}Basis Function\\ Independency \end{tabular} \\ \hline
SOS-based method \cite{jagtap2020formal}
&  \xmark     
&  \cmark 
&  \cmark               
&  \xmark           
&  \xmark
&  \cmark    \\ \hline
Scenario Optimization method \cite{nejati2023formal} 
&  \cmark                                     
&  \cmark                                     
&  \cmark
&  \xmark
&  \cmark 
&  \xmark      \\ \hline
Neural network method \cite{anand2023formally}  
&  \cmark                                        
&  \xmark                                      
&  \xmark
&  \xmark
&  \cmark 
&  \cmark\\ \hline
\textbf{ELM method (ours)} 
&  \cmark                                        
&  \cmark                                        
&  \cmark
&  \cmark
&  \cmark 
&  \cmark   \\ \hline
\end{tabular}
}
\end{table*}

The goal is to verify using some barrier certificate that if the system starts within $\X_0$, it will be safe. To do so, we have considered the barrier certificate as ELM $B(x; \beta)$. For synthesizing the certificate, we have considered $\varepsilon = 0.002759$ and set the hidden neuron number as $p=60$. We deploy Algorithm~\ref{algo:NN_training} to train the ELM and using $\eta = -0.011$, the network is assured to be formally verified with Lipschitz constant $\mathcal{L}_{barr} = 0.2209$ as $\eta + \mathcal{L}\varepsilon = -0.011 + 0.2209\times(5.964+1)\times0.002759 = -0.00675 < 0$. A set of sample trajectories of the controlled system with corresponding barrier values over the trajectories with time is shown in Figure \ref{fig:maglev}, where the blue dashed line marks the boundary of the initial set and the red dashed line marks the boundary of the unsafe set. One can see that the barrier value is always negative over the trajectories, satisfying the conditions. 

\begin{figure}[h]
\centering
\begin{minipage}{0.48\textwidth}
\centering
\captionof{table}{Quantitative Comparison of Different Algorithms}
\label{tab:computation_ncbf}
\resizebox{\linewidth}{!}{%
\begin{tabular}{|c|c|c|c|}
\hline
Methods &
  \begin{tabular}[c]{@{}c@{}}\textbf{Proposed}\\ \textbf{Method}\end{tabular} &
  \begin{tabular}[c]{@{}c@{}}NCBF with\\ Lipschitz \cite{anand2023formally}\end{tabular} &
  \begin{tabular}[c]{@{}c@{}}NCBF with \\ SMT solver\cite{zhao2020synthesizing} \end{tabular} \\ \hline
\begin{tabular}[c]{@{}c@{}}Average\\ Convergence\\ Time (In sec.)\end{tabular}       & $21.204 \pm 0.212$ & $2072.304 \pm 42.784$ & $2118.34 \pm 35.631$ \\ \hline
\begin{tabular}[c]{@{}c@{}}Lipschitz \\ constant\end{tabular}                        & $0.033$            & $1.5$                 & -                  \\ \hline
\begin{tabular}[c]{@{}c@{}}No. of Samples\\ required for\\ Verification\end{tabular} & $36\times10^4$     & $49\times10^6$        & -                  \\ \hline
\end{tabular}
}
\end{minipage}%
\hfill
\begin{minipage}{0.48\textwidth}
\centering
    \includegraphics[width=0.95\linewidth]{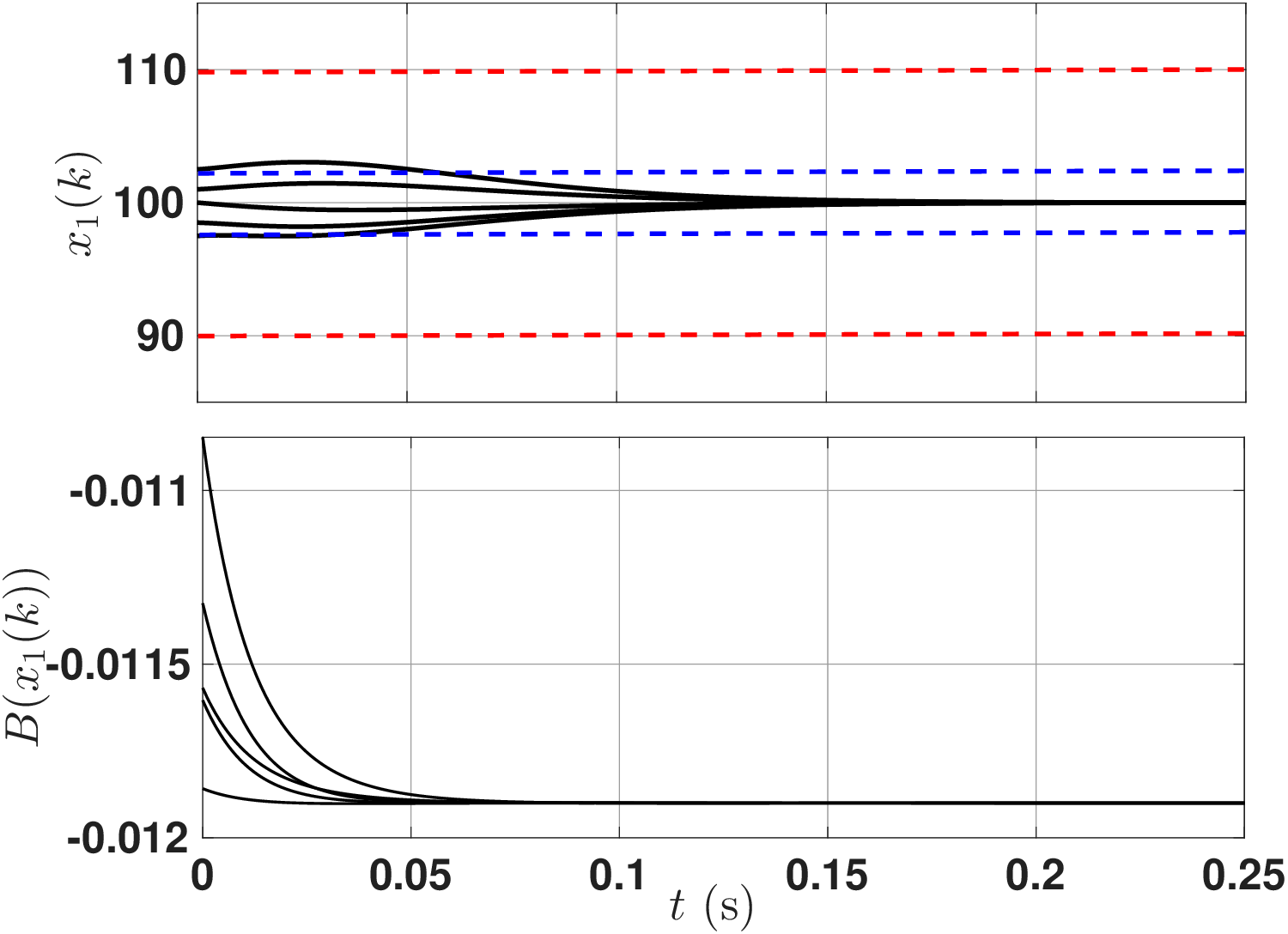}
    \caption{Magnetic Levitation System: Top- Trajectories of the ball position under stabilizing controller starting from different initial conditions, Bottom- Barrier values over time.}
    \label{fig:maglev}
\end{minipage}
\end{figure}

\section{Comparison with Related Approaches}

\subsection{Theoretical Comparison}
{We compare our ELM-based method with several existing methods that propose the synthesis of barrier certificates, such as dynamics-dependent conventional SOS-based techniques \cite{parrilo2003semidefinite} and dynamics-independent data-driven scenario optimization techniques \cite{nejati2023formal} and neural network-based methods \cite{anand2023formally}. While SOS-based techniques can synthesize the barrier certificate using convex programming, they depend on knowledge of dynamics, suffer from scalability due to the use of semidefinite programming and can not be applied to any general nonlinear dynamics.} The data-driven scenario optimization-based method can handle general unknown nonlinear dynamics but suffers from computational problems due to dependency on the basis function and not using convex programming. Neural network-based methods also suffer from a similar problem due to backpropagation of loss. Compared to the methods, our ELM-based approach can be applied to any general unknown nonlinear dynamics with less computation time, due to convex programming and can also be applied for higher-dimensional systems as well. The theoretical comparison is summarized in Table~\ref{tab:quality-elm}.

\subsection{Quantitative Comparison}

To substantiate the claim of reduction of the computation time, Lipschitz constant and efficient use of the sample collection procedure compared to classical neural barrier certificate synthesis methods, we compare our ELM-based technique to synthesize a barrier certificate with other existing neural CBF methods such as \cite{anand2023formally} and \cite{zhao2020synthesizing}. While \cite{anand2023formally} explores Lipschitz-based continuity to synthesize and verify the certificate, \cite{zhao2020synthesizing} uses an SMT-based solver to verify the trained certificates. We empirically show that the computation time is less for ELM compared to other methods, where each method has been tested for $10$ trials with random initializations. Furthermore, in \cite{anand2023formally} the Lipschitz constant of the barrier certificate is predefined, while the proposed method remove the choice and minimize it using the QP. As a result, due to the low Lipschitz constant, a smaller number of samples are required to formally verify the barrier in the case of ELM compared to \cite{anand2023formally}, proving the claim of sample-efficient algorithm synthesis for safety verification of an unknown system. Table~\ref{tab:computation_ncbf} reports the summary of the quantitative comparisons.   


\section{Conclusion and Future Work}
The work proposes a barrier certificate synthesis strategy using an extreme learning machine (ELM). Formulating the barrier conditions as a convex formulation, we deploy ELM to compute the barrier certificate. We use a Lipschitz-based condition and provide a post-hoc validity guarantee of the trained barrier. We validate our approach using two different case studies, and showcase the benefits of the approach through benchmark comparisons with existing neural barrier certificate synthesis approaches. Future work includes the controller synthesis procedure with an extension to interconnected systems.

\bibliographystyle{plain}
\bibliography{sources}

\end{document}